\documentclass[11pt]{article}

\usepackage[letterpaper,top=1.5in,bottom=1in,left=1.25in,right=1.25in]{geometry}

\usepackage[utf8]{inputenc} 
\usepackage[T1]{fontenc}    
\usepackage{hyperref}       
\usepackage{url}            
\usepackage{booktabs}       
\usepackage{amsfonts}       
\usepackage{nicefrac}       
\usepackage{microtype}      
\usepackage{xcolor}         

\usepackage{amssymb}
\definecolor{darkred}{rgb}{0.7, 0.0, 0.0}
\definecolor{darkgreen}{rgb}{0.0, 0.7, 0.0}

\definecolor{grafhed}{rgb}{0.0,0.545,0.545}
\providecommand{\graf}[1]{}

\usepackage{amsmath}
\usepackage{amsthm}

\newcommand{\E}{\mathbb{E}}
\newcommand{\Mid}{\,\middle|\,}  

\usepackage{thmtools}
\usepackage{thm-restate}

\newtheorem{lemma}{Lemma}

\newtheorem{assump}{Assumption}
\newtheorem{definition}{Definition}

\title{Disinformation, Stochastic Harm, and Costly Effort: A Principal-Agent Analysis of Regulating Social Media Platforms}

\author{%
  Shehroze Khan and James R.\ Wright
\\
  Department of Computing Science \\
  Alberta Machine Intelligence Institute (Amii)\\
  University of Alberta\\
  \texttt{\{shehroze,james.wright\}@ualberta.ca} \\
}

\usepackage[style=numeric]{biblatex}
\begin{document}

\maketitle

\begin{abstract}
    The spread of disinformation on social media platforms is harmful to society.
    This harm may manifest as a gradual degradation of public discourse; but it can also take the form of sudden dramatic events such as the 2021 insurrection on Capitol Hill.
    The platforms themselves are in the best position to prevent the spread of disinformation, as they have the best access to relevant data and the expertise to use it.
    However, mitigating disinformation is costly, not only for implementing detection algorithms or employing manual effort, but also because limiting such highly viral content impacts user engagement and thus potential advertising revenue.
    Since the costs of harmful content are borne by other entities, the platform will therefore have no incentive to exercise the socially-optimal level of effort.
    This problem is similar to that of environmental regulation, in which the costs of adverse events are not directly borne by a firm, the mitigation effort of a firm is not observable, and the causal link between a harmful consequence and a specific failure is difficult to prove.
    For environmental regulation, one solution is to perform costly monitoring to ensure that the firm takes adequate precautions according to a specified rule.
    However, a fixed rule for classifying disinformation becomes less effective over time, as bad actors can learn to sequentially and strategically bypass it.
    Encoding our domain as a Markov decision process, we demonstrate that no penalty based on a static rule, no matter how large, can incentivize optimal effort. Penalties based on an adaptive rule can incentivize optimal effort, but counterintuitively, only if the regulator sufficiently \emph{overreacts} to harmful events by requiring a greater-than-optimal level of effort. We offer novel insights for the effective regulation of social platforms, highlight inherent challenges, and discuss promising avenues for future work.
\end{abstract}

\section{Introduction}
\label{sec:intro}

Contemporary web and social media platforms provide a ripe ground for the spread of false news, hoaxes, and disinformation \cite{kumar2018false}.
Compounding the problem, social platforms' business models often conflict with efforts that can mitigate these problems. Facebook, for instance, uses machine learning models to maximize user engagement; however, in doing so, these models can also favor content that is toxic and filled with conspiracy, lies, and misleading, divisive information \cite{hao_2021, avaaz-blm, buzzfeed-us-elec}.

\graf{Disinformation definition and examples of harm}
We use \emph{disinformation} to refer to all such toxic content, including all kinds of false and fabricated news posing as truth, created with the intention to mislead \cite{kumar2018false}.
The unmitigated spread of disinformation is harmful to society. The harm can be direct physical or emotional distress to an individual; it can also manifest as a negative externality affecting public discourse, or social welfare. Examples include the undermining of public health response due to Covid-19 false rumors \cite{who-covid,avaaz-covid}, disease outbreaks due to anti-vaccination propaganda \cite{wp-antivaxx}, violent conspiracy movements surrounding the 2020 US presidential elections \cite{avaaz-us-elec,fb-capitol-attack}, and horrific incidents such as the Pizzagate shooting \cite{wp-pizzagate} and ethnic violence in Myanmar \cite{nyt-myanmar}.

\graf{Stochastic harm means platforms have no incentive to self-regulate}
The costs of these rare and dramatic events are borne exclusively by society, rather than the social platforms themselves. These events are inherently stochastic as it is impossible to predict with certainty that a given collection of content will cause a specific harm. Furthermore, deploying techniques to prevent the spread of disinformation is costly: filtering, demoting or assigning warning labels to associated content comprises both the direct costs of implementing classification algorithms or employing manual detection effort, and also the indirect opportunity costs of advertising revenue due to subsequent losses in user engagement \cite{nyt-fb, buzzfeed-fb-scam-profit}. Platforms such as Facebook and Twitter face no compelling incentives to prevent the spread of disinformation; thus, relying on platforms to police themselves will not work \cite{frank_law}. The only reason for a profit-motivated platform to control the spread of disinformation is to avoid penalties imposed either by users or a public regulator.

\graf{Incentive issue feeds into the technological issue}
Predicting whether a specific piece of content will cause some real-life harm is extremely hard. Indeed, disinformation remains a problem due to the inherent challenge of developing the technological tools required to effectively detect and mitigate it. Yet, solving this technical problem alone will not be sufficient in controlling the harm from the spread of disinformation. It is equally crucial to solve the incentive problem described above, as this will likely pose as an obstacle to solving the technical problem: if social platforms profit from the virality of disinformation due to their engagement-centric business model, and if platforms face no direct consequences for any resulting harm, platforms will not be compelled to solve the technical problem of disinformation \cite{disinfo-profit}. The incentive problem feeds back into and exacerbates the technological challenges of mitigating disinformation on social platforms.

\graf{What is different about this problem? Describe AI...}
Techniques for mitigating disinformation must leverage tools in artificial intelligence (AI), which further complicates the issue of misaligned incentives. The sheer scale at which users generate and share content on social platforms mean that any form of content moderation must rely, to some degree, on the automation afforded by AI in order to handle the vast volume of data. This aspect is different from traditional publishing, television, and print media where humans are involved in the editorial feedback loop before any content is allowed to be published. The problem of assigning liability for content is therefore much simpler in traditional media. With user-generated content on social platforms, however, the quality and accessibility of data determine whether AI will be effective at moderating content. But because only platforms have full, real-time access to their data, the problem is thus to motivate their use of AI to proactively mitigate disinformation --- in spite of their self-interest in not doing so \cite{frank_law} --- without having the same expertise or access to data.

\graf{Principal-agent models for disinformation}
The \emph{principal-agent framework} of microeconomics models the interactions between an \emph{agent}, who can influence the probability of an outcome by incurring costly effort, and a \emph{principal}, who has preferences over the outcome. We model the domain of disinformation prevention through this lens, with the platform as an agent who has the ability but not the incentive to undertake costly precautions against the spread of disinformation, and a regulator as the principal who seeks to balance the cost of the precautions against the harm caused by disinformation.

\graf{Key assumptions}
We begin by reviewing related work on techniques for mitigating disinformation and background on the principal-agent approach to modeling the regulation of stochastic externalities in Section~\ref{sec:rel-work}.
We then lay out our modeling assumptions in Section~\ref{sec:model-bg}.
We make three key assumptions. First, any attempt to regulate the actions of a platform before harm occurs (using a so-called negligence standard) must specify what level of effort for mitigating disinformation is adequate. Even if this specification is left implicit, we can model it as if it were an explicit public standard, set by the regulator, requiring some level of effort from the platform. Second, any given public standard will, in practice, require less effort from the platform over time as disinformation authors can learn to circumvent it; the platform is thus able to get away with expending less effort policing disinformation in order to save on costs, since the data and expertise needed to continuously re-train a model of content harmfulness is possessed by the platforms but not the regulators. Third, the public standard for content that ought to be prohibited on the platform will increase after a harmful event.

\graf{Descriptive MDP and our formal results}
We formalize these assumptions as a Markov decision process (MDP) in Section~\ref{sec:formal-model}, and use the model to derive our main results. We show that that no level of fines based on a fixed public standard can induce the socially-optimal effort for mitigating disinformation. However, in the presence of a public standard that reacts to a harmful event via an increase in the required level of effort, the platform's individually-optimal level of effort may exceed the level currently required by the public standard. In particular, the platform may be incentivized to continue exerting effort at a specific threshold when the public standard becomes less stringent over time. However, perhaps counterintuitively, this effort threshold will fall short of the socially-optimal level unless the public standard sufficiently \emph{overreacts} as a response to any harmful event, by requiring a level of effort that is greater than socially optimal.

\graf{Impossibility result with new MDP setup}
Finally, to further demonstrate the complexity of this incentive problem, we show that even under a simpler, more stylized setting --- where the regulator has the same technical ability as the platform and the costs of harm from disinformation are known --- absent knowledge of the platform's costs of effort, there is no specification of the public standard's required effort that will always incentivize the socially-optimal level of effort.
We thus conclude that the design of mechanisms that may elicit the costs of foregone engagement incurred by the platform in policing disinformation is one of the promising directions for future work.

\section{Background and Related Work}
\label{sec:rel-work}

\subsection{Fighting Falsity Online}
    \paragraph{Detecting disinformation.} A popular approach towards limiting online disinformation is to develop tools or frameworks that are effective in detecting associated content. This process aims to identify disinformation in its initial stages so that mitigating efforts thereafter may restrict or eliminate exposure to users of social media. \textcite{fake-news-detection-survey} survey some techniques that make false news detection efficient and explainable. These techniques are categorized into four areas: knowledge-based methods that involve fact-checking, style-based methods that focus on studying linguistic features of false content, propagation-based methods that analyze how such content spreads in the social network, and source-base methods that investigate the credibility of sources that generate false news. The goal of studying these and other characteristic features of the false news ecosystem, such as those surveyed by \textcite{kumar2018false}, is to develop algorithms and tools for early detection.
    
    Knowledge-based methods mainly involve fact-checking, which in turn can be either manual or automatic. Fact-checking is the process of extracting claims made in a given piece of content that is to be verified and checking these against known facts \cite{fake-news-detection-survey}. Manual fact-checking can either be crowd-sourced from users on social platforms, similar to Facebook or Twitter provisioning its users with the ability to report hoax content \cite{fb-fact-checking,twitter-fact-checking}; or it can also be conducted via third-party websites such as \textit{Snopes}\footnote{\url{https://www.snopes.com}}, \textit{PolitiFact}\footnote{\url{https://www.politifact.com}}, or \textit{FactCheck}\footnote{\url{https://www.factcheck.org}} that employ domain experts dedicated to serving the public by debunking disinformation.
    
    While knowledge-based and style-based detection techniques focus on analyzing the textual content of disinformation --- so that predictive classifiers might be trained to readily and effectively flag false news --- propagation-based techniques study how such content disseminates amongst users in a given social network. \textcite{twitter-false-news-spread} have conducted an empirical study of tweets on Twitter to analyze the differences between the spread of true and false news stories. It is shown that false news stories travel faster, farther, and more widely than true news stories \cite{twitter-false-news-spread}. Such an analysis not only aids in investigating the causes and consequences of disinformation proliferation, but it also helps formulate propagation-based false news detection as a classification problem.
    
    Our work starts from the assumption that the platform has the ability to detect and limit the spread of objectionable content \cite{fb-ai}. Our focus instead is on modeling the incentives faced by the platform to \emph{not exercise this ability.}
    
    \paragraph{Mitigating the effects of disinformation.} Once effective technology for detecting disinformation content and diffusion networks is implemented, the next step is to mitigate or limit the impact such content may have on users. A straightforward approach is to simply remove associated content from the platform entirely; another is to demote or down-rank content so that it is less likely to be served on users' feeds. These approaches rely on platforms to undertake action to reduce the spread of disinformation since the recommendation algorithms serving content to users are proprietary. However, there are also studies conducted by third-party researchers offering other solutions for mitigating the effects of disinformation once it has entered the social network.
    
    One such technique draws from concepts in human cognitive psychology to study deception cues that influence users' decision-making process related to sharing content in the social network \cite{cogsci-detection}. The goal here is to provide users with informative cues so that they are less likely to share disinformation.
    Another, more proactive intervention is the ``Facts Before Rumors'' campaign \cite{fax-b4-rum}, where the focus is to preempt the kinds of rumours that are likely to spread on a social network --- based on user locations and localized news content, for example --- and counteract these in advance by employing certain users to spread truthful news.
    Other interventions focus on curing the effects of disinformation instead of preventing it initially. For example, the ``Correct the Record" initiative proposes a visual correction that may be sent to users exposed to false content on Facebook \cite{avaaz-ctr}.
    
    Again, our work focuses less on advocating a particular mitigation technique; rather, our goal is to determine the conditions under which platforms can be induced to actively implement any such technique to prevent harm from the spread of disinformation.
    
    \subsection{Hidden-Action Principal-Agent Model}
    Many economic interactions involve two parties, a principal and an agent, where the agent's choice of action imposes some form of (negative or positive) externality on the principal. In most realistic scenarios, the principal cannot directly monitor or observe the agent's action, but instead only observes a stochastic outcome resulting from it. For example, in the interaction between a property insurer (principal) and a property owner (agent), if the insurer bears the costs of any damages to the property, the owner might not be incentivized to maintain it and might engage in risky behaviors (e.g., leave the kitchen unattended while cooking). This situation exemplifies the problem of \textit{moral hazard}, which is an important feature of the principal-agent interaction because it precludes straightforward incentive schemes. Many employment settings also share this characteristic. For example, the CEO (principal) of a small startup company --- whose income is directly related to the company's growth and product sales --- would want their employees (agent) to undertake effort that profits the company (e.g., a UI/UX developer improving the company's website leading to increased traffic and sales). But if the employees are simply compensated at a fixed hourly rate, they might not be incentivized to put in their best effort to benefit the company.
    
    Naturally, it will be in the principal's interests to influence the agent's choice of action. The principal may therefore be invested in drafting a contract for such influence in order to guard against the problem of moral hazard \cite{inbal-contract-1}. The need for a contract arises due to information asymmetry between the two parties --- i.e., the agent has more information or expertise about their actions than the principal. For property insurance, the hidden information is the agent's act of not maintaining the property and engaging in some risky behavior; for the startup company example, the hidden information is the UI/UX developer's expertise in developing clean, functional websites; whereas for our setting, AI is the source of asymmetric information: only platforms possess the expertise, models and data to promptly flag and mitigate disinformation.
    
    \subsection{Contract Theory Meets Computer Science}
    The principal-agent model is central to \textit{contract theory}, which is an important field in microeconomics. This area has recently gained traction in the algorithmic game theory community, primarily through works such as \cite{inbal-contract-1,inbal-contract-2,duetting2021combinatorial}, where the aim is to concisely represent principal-agent settings and computationally characterize the design of \textit{optimal} contracts\footnote{An optimal contract is one that maximizes the principal's expected reward assuming that the agent best responds to the contract \cite{duetting2021combinatorial}.} permitted by such settings.
    
    Our work is similar to these studies in that we consider the optimal design problem of maximizing the utility of the principal, who in our setting is a social welfare-maximizing regulator. Yet, instead of a computational complexity analysis, we represent disinformation prevention as a principal-agent problem through our descriptive MDP model, which to our knowledge is a unique approach towards modeling the incentives faced by social platforms pertaining to the mitigation of false news and other toxic content. Unlike those cited works, the outcome space for our setting is simply the realization of a single harmful event due to the unmitigated spread of disinformation; our focus as such is specifically on the design of penalty contracts or schemes enforced by some regulatory agency in order to contain this stochastic externality, or harm from disinformation.
    
    \subsection{Regulating Stochastic Externalities}
    The hidden action principal-agent model can also be applied to the regulation of firms that generate stochastic externalities as a result of their operations. Examples include harmful accidents such as medical product failures, oil spills, nuclear waste leakages and other forms of pollution \cite{INNES1999181}. Moral hazard exists in these settings because firms (agent) might not be incentivized to take a costly precaution (unobservable action) to reduce accident risk, which is where a regulatory authority (principal) steps in to specify a penalty contract to guarantee some form of enforcement.
    
    \textcite{oilspill} explores optimal enforcement strategies for the regulation of firms that stochastically pollute the environment in the form of oil spills. It is shown that under a \emph{strict liability standard}, where a polluting firm is always penalized if an oil spill occurs regardless of its level of precautionary effort, the firm can be induced to exercise the socially-optimal or \textit{first-best} level of effort. However, this requires that a specific firm can be identified as being responsible for a spill.
    
    When a strict liability standard is impractical --- for example because the perpetrator of harm cannot be reliably identified --- a regulator might prefer to expend resources to monitor a firm's effort directly.
    In these situations, a \emph{negligence standard} can be preferable, in which a firm is not held responsible for an accident if it can demonstrate that it took adequate precautions.
    Naturally, the quality of information available for regulatory monitoring is a consideration for enforcing such a standard \cite{shavell-risk-sharing-pa}.
    
    Our domain shares many of the features of the oil spill prevention domain: there are stochastic externalities associated with the spread of certain kinds of content on social platforms (harm from disinformation), as there are with firms transporting oil (oil spills); the likelihood or severity of such harm may be reduced to some degree if platforms exercise responsible and proactive content moderation, but not completely eliminated as the harm is ultimately a direct outcome of individual actions --- akin to a spill that occurs because of inclement weather and not due to the oil tanker being faulty.
    
    However, our domain is also importantly different from that of oil spill regulation, or environmental regulation more generally. The following subsection expands on these differences.  In Section \ref{sec:model-bg} that introduces our formal, descriptive model, we will elaborate on the similarities and highlight how these key differences prevent the application of standard enforcement strategies.
    
    \subsection{Why Online Disinformation is Different}
    Disinformation prevention via regulatory mechanisms has its own unique challenges. First, there are ongoing debates around assigning liability for content hosted by social platforms \cite{frank_law}, particularly due to editorial control being different for the social media setting. As discussed previously, it is infeasible to implement human-in-the-loop feedback for every piece of real-time, user-generated content shared on online platforms, as this medium is unlike traditional forms of media; there exist as such not only the issue of scalabilty for any disinformation mitigation technology, but also the question about whether similar liability rules for harmful content should apply to social media as they would for traditional media.
    
    Second, in order to handle the vast volume of content, AI must be utilized for the proactive and automated flagging of disinformation. This aspect complicates regulation because the data powering such AI is only accessible to the social platforms themselves. Moreover, the recommendation algorithms that filter and serve content to users are also proprietary. Therefore, unlike for the environmental regulation domain, mandating exact precautions against the spread of disinformation for social platforms is likely to be an involved process for any regulatory authority --- especially in comparison to, for example, specifying precise conditions that render an oil tanker safe for the transport of oil, or promoting adequate technology that will reduce emissions causing air pollution.
    
    Third, and also different from pollution regulation, there exists the issue of malicious actors responding strategically to any explicitly fixed rules or precautions against the spread of disinformation. Authors and purveyors of disinformation are constantly coming up with new, sophisticated methods to ensure that their fabricated stories disseminate online: techniques include obfuscation strategies to hide disinformation propagating networks and the origins of propagandist content; and also changing the content itself via constructing new falsehoods, or targeting different groups \cite{evolving_disinfo, hao_2021}. Any successful attempts to moderate such users or content at scale must therefore utilize all the technical expertise and data required to counteract efforts of these bad actors. Regulation becomes challenging because only social platforms have access to such resources and data, and they are not necessarily incentivized to undertake action at the expense of losses in user engagement \cite{nyt-fb,frank_law}.
    
    \subsection{Mechanism Design}
    Another closely related body of work is the economic theory of mechanism design, where the goal is to design protocols or procedures that mediate interactions between strategic agents in order to achieve some desired objective. Naturally, the outcome is subject to the constraint that agents behave selfishly, in that they act according to their rational self-interests; and also that agents hold some private information, i.e., their \textit{hidden types}. A mechanism seeks to attain the desired outcome by incentivizing agents to report their private types. Mechanism design theory contrasts with the standard principal-agent model with respect to where the information asymmetry exists: it is the agents' type information that is hidden from the mechanism designer; whereas, for the principal-agent model, the principal cannot directly observe an agent's action(s), which form(s) the source of asymmetric information.
    
    Because our work is concerned with setting up a regulatory policy in order to achieve a desired social outcome --- that is, the socially-optimal level of control of disinformation --- mechanism design is a pertinent framework for our domain. Yet, it does not directly apply to our setting 
    since the regulator cannot reliably observe a social platform's efforts, or action, to curb the spread of harmful content. We therefore utilize the principal-agent framework to model the regulation of disinformation. The goal for the regulator (principal) is to incentivize a platform (agent) to use its proprietary expertise and AI technology --- which are not available to the regulator --- to responsibly limit toxic and harmful content in order to control the harm from disinformation. The following section introduces our formal descriptive model.

\section{Modeling the Regulation of Disinformation}
\label{sec:model-bg}
\graf{Domain similarity}

We have the following scenario: A regulator (principal) would like the platform (agent) to limit the amount of disinformation spread to control the likelihood of harm, which is a stochastic and observable event. The underlying assumption is that the unmitigated spread of disinformation on social platforms makes the occurrence of harm more likely.

\graf{The platform is able to detect potentially harmful content accurately}
We assume the platform possesses a proprietary classification model that accurately assigns for every a piece of content the probability of it causing harm \cite{fb-ai, nyt-fb}. Thus, extremely violent, graphic, or objectionable content, which contains nudity, racism, child pornography, or any form of human/animal abuse, is tagged by the model with a very high harm probability value. Other, benign forms of content, such as cute photos of pets or birthday greetings, are assigned with a very low harm probability value.

\graf{What we mean by effort}
We summarize all of the measures that a platform takes to mitigate harmful content such as disinformation as the ``effort'' expended by the platform. This effort includes both changing moderation rules and measures (such as automated detection of harmful content) for their enforcement. We model  effort by assuming that the platform picks a harm probability threshold (e.g., by specifying content moderation rules) that represents the platform's tolerance for hosted content --- that is, all content whose harm probability value exceeds this threshold is considered by the platform as being unacceptable and in violation of its community standards of acceptable postings.
Thus, the platform's effort includes both detecting  such content, and thereafter employing techniques to mitigate it (e.g., via enforcement of rules).

These mitigation techniques could include filtering content entirely, downgrading it so that it appears on fewer user feeds, or labeling it with a warning invoking users' discretion. As discussed previously, the exact choice of technique is not important for this analysis; any and all such methods effectively count as the platform exercising effort to prevent the spread of harmful content and, by extension, disinformation.

\graf{Interpreting Effort}
\paragraph{Interpreting Effort.} High effort can be interpreted as the platform proactively updating its rules to retrain its model for the automated detection of new forms of harmful content, and perhaps also employing manual effort in tandem to responsibly moderate and control the spread of such content. Conversely, low effort may be thought of as the platform being lax about enforcing its moderation rules, and perhaps even as not updating these rules to preempt the spread of toxic content. In the context of our modeling, therefore, a low choice of threshold implies stricter content moderation rules and thus high effort on part of the platform as more content items will be flagged; and a higher threshold indicates laxer content moderation rules and thus lower effort as fewer content items will be flagged by the platform.

\graf{The model and its various params}
Let $H$ be a binary random variable indicating whether harm occurs with density function $h(e)=\Pr[H \mid e] \in (0,1]$ representing the probability that harm occurs if the platform exerts effort $e$.
\graf{Harm probability} Similar to \cite{oilspill}, we assume that although the platform is unable to control this externality directly, the platform can make it less likely for harm to occur  by exercising more effort. In line with the standard economic model of unilateral accidents \cite{INNES200429,INNES1999181,unilateral-model-ambiguity}, we assume that there are diminishing returns to effort; that is, effort reduces risk of harm at a decreasing rate: $h'(e) < 0$ and $h''(e) \ge 0$.\footnote{Though the cited studies assume strict convexity of harm function, i.e., $h''(e) > 0$, our results are robust towards slightly relaxing this assumption.} 

\graf{Effort is costly}
The business model of most social platforms is primarily based on advertising. For instance, advertising accounted for 98\% of the Facebook's \$86 billion revenue in 2020 \cite{fb-revenue}. Essentially, platforms monetize users' attention by optimizing their engagement for content, a subset of which includes disinformation. Thus, in addition to direct costs, limiting disinformation is also costly for platforms in terms of these indirect costs of losing potential ad revenue. Let $c(e)$ denote the cost of exerting effort $e$. We assume effort is increasingly costly; i.e., $c'(e) > 0$ and $c''(e) > 0$, which is also standard under the unilateral accident model.

\graf{Policy goal}
Given their behavioral advertising business model, platforms face no incentives to moderate attention-grabbing content, toxic or otherwise, especially because they do not directly incur the costs of any societal harm \cite{nyt-fb, frank_law}. Under this scenario of misaligned incentives, a social welfare-maximizing regulator aims to incentivize the platform to exercise adequate precautions against the spread of disinformation. Concretely, the regulator wishes to maximize the expected social welfare,
\begin{equation}
    \label{eq:social-welfare}
     EW(e) =  -h(e)D-c(e),
\end{equation}
where $D$ is the cost of damages as a result of any harm due to disinformation, assumed to be constant here for simplicity.

\graf{What is e*}
The socially-optimal or first-best effort maximizing ($\ref{eq:social-welfare}$) is given by $e^* = \arg\max_e EW(e)$. At $e^*$, the sum of the total expected costs of harm, or $h(e^*)D$, and the platform's costs of exerting this effort, or $c(e^*)$, is minimized; thus, $e^*$ by definition is the platform's precautionary effort at which the cost of any additional effort is balanced by the expected cost of damages due to harm.

We now discuss possible methods by which the platform may exert effort $e^*$, and further expand on domain specific features for our descriptive model.

\subsection{Strict Liability}

Under the strict liability standard, the platform is held completely liable for any harmful event, irrespective of its precautionary effort. To incentivize the first-best level of effort $e^*$, the strict liability fines $T$ must equal $D$, the societal cost of harm \cite{oilspill}; thus, the platform's expected utility is given by,
\begin{equation}
    \label{eq:fb-sl}
    EU(e) = -c(e)-h(e)T,
\end{equation}
which equals the expected social welfare equation (\ref{eq:social-welfare}).

\graf{Pros and Cons of SL}
A regulator might pick this enforcement standard because it does not require expending resources to monitor the platform's effort, which is only imperfectly observable because of the difficulty in identifying the exact mechanics of the platform's proprietary algorithms. A strict liability standard is effective for controlling the spread of unambiguously harmful content such as child pornography, since such content directly constitutes harm. However, strict liability for regulating disinformation might be impractical for a few reasons.

Most importantly, the direct causal links between any harmful event due to the spread of disinformation and the platform are sufficiently loose for this standard not to work, since the perpetrators are ultimately individuals; the platform can claim plausible deniability, or point to efforts at prohibiting dangerous content after the harm has already occurred, akin to when Facebook and Twitter banned groups like "QAnon" or "Proud Boys" after the insurrection on Capitol Hill \cite{cnbc-fb-coo}. Essentially, unlike explicitly toxic content such as child pornography, it is difficult to determine what constitutes direct harm for disinformation. Furthermore and relatedly, strict liability is unlikely to work in practice because it is also difficult to estimate $D$ a priori, as this harm could manifest in different forms.

\subsection{Negligence}

Under this standard, a regulator must specify a duty of care that the platform must follow in order to avoid liability for any harm. Monitoring the platform's effort is thus necessary to determine liability.

\graf{Monitoring and descriptive case}
Although monitoring is imperfect, the platform's content moderation efforts are not completely unobservable: there exists a crude public notion about the kinds of content that ought to be limited on social platforms. From an incentive standpoint, a negligence standard already exists in the sense that there is not a lot of nudity or child pornography, or content with explicit death threats, vile or racist remarks on most social platforms --- platforms like Facebook and Twitter expend ample resources to enforce their community standards via active content moderation \cite{fb-ai, fb-comm-std}. Presumably, platforms do not want public outrage, or to be charged with trafficking or any other forms of liability for such content, which if not controlled would be reported extensively in popular press.

\graf{Public model}
We model this descriptive situation with the presence of an explicit public standard, operated by a regulator, that fixes a required level of precautionary effort for mitigating disinformation. In reality, there is no concept of an explicit public standard specifying effort, but rather an implicit public notion about the types of content that ought to be moderated by the platform. Nonetheless, regardless of what the public standards for content are at any given moment, these standards imply a certain level of precautionary effort, which we encode with the presence of an explicit standard to simplify our formal analysis.

\subsection{Performative Prediction of Disinformation}
\label{sec:performativity}
When predictions about the actions of an agent change the outcomes for that agent, there is a risk that the predictive model will cease to be accurate \cite{perdomo2020performative}. We say that predictions that exhibit this problem are \emph{performative}.\footnote{Note that
we use the terms performative and performativity in a specific, strictly technical sense that differs from their colloquial usage.}
For example, a certain keyword that is extremely predictive of a message being spam may cease to be predictive once we filter based on it, as spammers will now have an incentive to stop using that keyword. Classifying disinformation is performative in this sense because bad actors can learn to bypass any detection model with new forms of disinformation \cite{evolving_disinfo, hao_2021}.

\graf{Justification for why the platform is better equipped than any regulator to detect and limit disinformation}
We assume the platform has sufficient technical resources and the data to retrain its proprietary model in order to counterbalance performativity; i.e., the platform is able to successfully classify future modifications of disinformation via predicting true harm probabilities of associated content. The same is not true for the regulator-specified public standard as the regulator does not possess the same expertise or access to data. The regulator in theory could utilize open-source, state-of-the-art disinformation detection learning models to effectively flag false content as not satisfying the public standard \cite{fake-news-detection1, fake-news-detection2}. Yet, to the extent that platform data is not completely accessible to the public \cite{fb-api-restrictions}, these open-source models will be susceptible to performative prediction of new, evolved forms of disinformation unless retrained with the same, easily accessible data that is available to the platform.

\graf{Decaying public model and ex post backlash}
Consequently, because it is publicly accessible, the public standard weakens over time due to performativity as disinformation authors strategically learn to circumvent it. We encode this feature effectively as a gradual \textit{downward drift} or decrease in the public standard's required effort if harm does not occur. However, if harm occurs, we see a public backlash in that the public's tolerance of content linked to the harmful event gets lower ex post. This is akin to when Facebook and Twitter began suspending accounts, content, and hashtags linked to the Capitol Hill riots \cite{fb-ex-post-filtering-qanon}. We encode this backlash as an effective increase in the public standard's required effort as a response to a harmful event.

\section{Formal Model}
\label{sec:formal-model}
We formalize our model as a MDP incorporating descriptive features of our domain as described in the previous section and defined by $(S,A,P_e,R_e)$ where
    $S$ is the discrete state space of the current effort $e_c$ as specified by the public standard,
    $A$ is the continuous set of actions representing the platform's choice of effort $e$,
    $P_e(e_c, e_c') = \Pr[s_{t+1}=e_c' \mid s_t=e_c,a_t=e]$ is the transition probability to state $e_c'$ by exerting effort $e$ in state $e_c$, and 
    $R_e=-c(e)$ is the immediate reward of exerting effort $e$, which is simply the cost of effort $e$.
    
    Consistent with MDP literature \cite{Sutton1998}, we use $\pi : S \rightarrow A$ to denote an arbitrary, deterministic policy specifying the platform's choice of effort $e \in A$ for all $e_c \in S$.
    The \emph{state value function} $v_{\pi}(e_c) = R_e + \gamma \E[v_{\pi}(s_{t+1})]$ is the expected discounted value of following policy $\pi$ from state $e_c$; the \emph{state-action value function} $q_{\pi}(e_c, e)=R_e + \E[v_{\pi}(s_{t+1})|A_t=e]$ is the expected discounted value of choosing effort $e$ in state $e_c$, and then following policy $\pi$ thereafter.

\subsection{Optimal Effort Under a Fixed Public Standard}
\label{sec:static-public}
In the first analysis, we assume no downward drift of the public standard's required effort level and no backlash if harm occurs; i.e., the effort required by the public standard remains fixed at $e_c$. Under this negligence standard, the platform is only subject to ex ante regulation via regulatory audits, and not penalized ex post if harm occurs.

\graf{Ex ante regulation}
Let $r \in [0,1]$ be the probability that the regulator conducts an audit of the platform's effort and let $P_f(e  \mid  e_c) \in [0,1]$ be the probability that the platform fails its audit if it exerts effort $e$, given the current required effort $e_c$. If the platform fails the audit, it is liable for negligence fines $F$. Thus, assuming risk-neutrality, the platform's expected utility under the regulatory regime of ex ante negligence is,
\begin{equation}
    \label{eq:fb-ex-ante-reg}
    EU(e \mid e_c) = -c(e) -rP_f(e \mid e_c)F.
\end{equation}

\begin{definition}
\label{def:adequate-effort}
The adequate level of effort $e$ is the point beyond which the probability of failing the audit $P_f(e  \mid  e_c)=0$, where $e_c$ is the effort as prescribed by the public standard.
\end{definition}


There exists some level of effort $e$ that guarantees that the platform will not incur any fines for negligence; $e$ will therefore be considered adequate from the perspective of the regulator, since the platform not incurring fines implies that it does not fail the audit by exerting effort $e$. Note that $e_c$ is automatically deemed adequate because it is specified by the regulator. Thus, by definition, $P_f(e  \mid  e_c)=0$ for all $e \ge e_c$: the platform can guarantee that it will not incur fines by exerting at least effort $e_c$ in state $e_c$ for all $e_c \in S$, and so the platform never fails its audit by fully complying with the explicit public standard.

By inspection, it is trivial to see that the platform's individually-optimal level of effort will never exceed $e_c$ for all $e_c \in S$, irrespective of how large the size of negligence fines $F$ is: for any two adequate effort levels, the platform will exercise lower effort because that will maximize (\ref{eq:fb-ex-ante-reg}). The following proposition formalizes this claim.

\begin{restatable}{prop}{secondprop}
\label{prop:bare-minimum-effort}
Given a fixed adequate effort level $e'$, there exists no fine scheme $F$ that can incentivize the platform to exert more effort than $e'$.
\end{restatable}

\textbf{All proofs are deferred to the appendix.}


If the public standard remains fixed at $e_c$, no amount of fines solely based on ex ante negligence regulation, no matter how large, can induce the platform to exercise more than $e_c$ effort in state $e_c$. Only if $e_c = e^*$, and if the regulator can guarantee full compliance with the public standard, can this scheme incentivize socially-optimal effort. Notice, however, that in reality the public standard for content will weaken over time due to performative prediction. Thus, what is considered adequate effort by the regulator will change for different states of the MDP, corresponding to different values of $e_c$; consequently, even if the regulator-specified $e_c$ magically happens to equal $e^*$, it is not bound to stay at $e^*$ indefinitely. Clearly then, regulation only via ex ante negligence will fail to induce the platform to exert optimal effort. The subsequent analysis illustrates how the presence of a public backlash incentivizes the platform to exert more effort than explicitly required by the regulator.

\subsection{Optimal Effort Under an Adaptive Public Standard}

We now consider an adaptive MDP setup. The current state represents the required level of effort $e_c$; if no harm occurs, the required effort reduces over time due to performativity; and if harm does occur, then the required effort increases to $e_h$, representing public backlash.

\begin{assump}
\label{assump:suff-fines}
Given a fixed fine structure $F$ and the effort required by the public standard $e_c$, the platform's individually-optimal effort level is at least $e_c$.
\end{assump}

This assumption is without loss of generality as we will label the states of the MDP according to the platform's individually-optimal static effort.\footnote{The platform will thus never exert less than $e_c$ effort in state $e_c$, but we will see that it will sometimes exert more.} Essentially, the goal is to determine if the platform can be induced to exert more than $e_c$ effort in state $e_c$.

\begin{assump}
\label{assump:e_h-transition}
At state $e_c$, the transition probability to the high effort state $e_h > e_c$ is simply $P_e(e_c, e_h) = h(e)$, the probability that harm occurs given the platform exercises effort $e$.
\end{assump}

Note that the harm probability and thus the transition to state $e_h$ only depends on the platform's effort $e$, and not on the state $e_c$. This transition encodes the public backlash.

\begin{definition}
\label{def:next-lower-state}
The next state with required effort lower than $e_c$ is $\chi(e_c) = \sup \{ s \in S \mid s < e_c \}$.
\end{definition}

\begin{assump}
\label{assump:cont-drift}
If harm does not occur, we assume a continuous downward drift of the public standard's prescribed effort --- that is, the effort either lowers to $\chi(e_c)$ with drift probability $g(e_c)$, or stays fixed at $e_c$ with probability $1-g(e_c)$.
\end{assump}

Note that the drift probability to state $\chi(e_c)$ only depends on the current state $e_c$, and not the platform's effort $e$, conditional on the harm's not occurring.
The decrease in effort encodes performativity.


\begin{restatable}{lemma}{firstlemma}
\label{lem:aggressive-filtering}
    Fix a state $e_c$ representing the current effort required by the public standard, and an arbitrary policy $\pi$, and let $e_h > e_c$ be the effort that the public standard will require if harm occurs.
    For all $e_2 > e_1 \ge e_c$,
    \begin{equation}
        \label{eq:thresh-preference}
        q_{\pi}(e_c, e_2) > q_{\pi}(e_c, e_1) \iff
        d(\pi, e_c) - v_{\pi}(e_h) > \frac{c(e_2) - c(e_1)}{\gamma (h(e_1) - h(e_2))},
    \end{equation}
    where $d(\pi, e_c) = g(e_c)v_{\pi}(\chi(e_c)) + (1-g(e_c))v_{\pi}(e_c)$.
\end{restatable}

Given $e_c$, Lemma \ref{lem:aggressive-filtering} specifies the condition under which the platform's picks one effort level over another from the continuous action set $A$, expressed via the state-action value function of the MDP.

\begin{definition}
\label{def:thresh-strategy}
We call $\pi^\tau$ a \emph{threshold strategy with threshold $\tau$} if $\pi^\tau(e_c) = \max\{\tau,e_c\}$ for all $e_c \in S$.
\end{definition}

Threshold strategies form a class of policies that can induce more aggressive effort as specified by the condition in Lemma \ref{lem:aggressive-filtering}. The following results characterize important features of threshold strategies lending support to our main derivation of the platform's optimal policy in Theorem \ref{thm:form-optimal-policy}.

\begin{restatable}{lemma}{secondlemma}
\label{lem:tau-states}
Given a threshold strategy $\pi^\tau$, the state value function $v_{\pi^{\tau}}(e_c)$ is fixed for all $e_c \le \tau$.
\end{restatable}
Lemma \ref{lem:tau-states} fixes the reward of exerting effort at a specific threshold, thereby enabling a straightforward characterization and analysis of the platform's policy --- amid all the possible drifting states of the public standard --- by means of a stable level of effort $\tau$.

\begin{restatable}{prop}{thirdprop}
\label{prop:e_h-worst-state}
For all threshold strategies $\pi^\tau$, we have that $v_{\pi^{\tau}}(e_h) \le v_{\pi^{\tau}}(e_c)$ holds for all $e_c \in S$.
\end{restatable}

Proposition \ref{prop:e_h-worst-state} establishes $e_h$ as the worst state for the platform following a threshold strategy. Intuitively, because it encodes the public backlash, $e_h$ by definition is the highest effort the public standard will require and thus it must also yield the lowest expected discounted reward for the platform. Crucially, this guarantee of the lowest reward in state $e_h$ acts as the incentivizing mechanism for the platform to exert more effort than explicitly required by the public standard.

Given these MDP dynamics of performativity and public backlash as a response to harm, we characterize the platform's individually-optimal effort policy at any state $e_c$. The following existing results support our main result in Theorem \ref{thm:form-optimal-policy}.

\begin{lemma}[\cite{boyd_vandenberghe_2004}]
\label{lem:fo-condition-convexity}
    Suppose $f$ is a differentiable function of one variable in dom($f$). Then $f$ is convex if and only if
    \begin{align*}
        f(y) - f(x) &\ge f'(x)(y-x)
    \end{align*}
    holds for all $x, y \in$ dom($f$). And analogously for strict convexity, 
    \begin{equation}
        \label{eq:strict-convexity-condition}
        f(y) - f(x) > f'(x)(y-x)
    \end{equation}
    for all $x \ne y$.
\end{lemma}

\begin{lemma}[\cite{Sutton1998}]
\label{lem:policy-improvement}
    Given a pair of deterministic policies $\pi$ and $\pi'$ such that for all states $s \in S$
    \[ q_{\pi}(s, \pi'(s)) \ge v_{\pi}(s), \]
    then $v_{\pi'}(s) \ge v_{\pi}(s)$.
\end{lemma}

We now specify the platform's individually-optimal policy under the adaptive MDP setup: the public standard's prescribed effort $e_c$ increases to the high effort state $e_h$ if harm occurs; and $e_c$ decreases over time conditional on harm not occurring. The following theorem demonstrates that the optimal policy for the platform under these dynamics is to follow a threshold strategy.

\begin{restatable}{thm}{mainthm}
\label{thm:form-optimal-policy}
    The optimal strategy $\pi^*$ for the platform is a threshold strategy $\pi^* = \pi^{\hat{e}}$, with threshold
    \begin{equation}
        \label{eq:e_hat-def}
        \hat{e} = \sup\left\{e \in [0,e_h] \Mid
        q_{\pi^e}(s^{-1}(e), e) - q_{\pi^e}(e_h, \pi^e(e_h)) \ge -\frac{c'(e)}{\gamma h'(e)} \right \},
    \end{equation}
    where $s^{-1}(e) = \sup\{e_c \in S \mid e_c \le e\}$.
\end{restatable}

The result follows from the first-order condition of convexity \cite{boyd_vandenberghe_2004} and the policy improvement theorem \cite{Sutton1998}. Intuitively, the theorem statement holds because past a certain level of effort, the gain to the platform of not exerting more effort is traded off against the increased probability of transitioning to the $e_h$ state, which yields the lowest expected reward as shown in Proposition \ref{prop:e_h-worst-state}.

The primary takeaway from Theorem \ref{thm:form-optimal-policy} is that the platform is incentivized to exert more aggressive effort at threshold $\hat{e}$, despite an over-time reduction of the public standard's prescribed effort $e_c$ due to the performative prediction of disinformation. Thus, the platform's optimal effort level is stable at $\hat{e}$ for all states $e_c \le \hat{e}$. The regulatory scheme that induces more aggressive effort is the ex post public backlash, i.e., when the required effort increases to $e_h$, which effectively poses as stricter future ex ante regulation as a response to a harmful event.

This result is also important because with the correct choice of public backlash $e_h$, the platform can in theory be induced to exert the socially optimal level of effort $e^*$. We formalize this claim in the following proposition.

\begin{restatable}{prop}{existenceprop}
\label{prop:target-thresh}
For any given socially optimal level of effort $e^*$, there exists a MDP consistent with our given conditions such that the optimal policy for the platform is a threshold strategy with threshold $\tau=e^*$.
\end{restatable}
The existence proof for $e_h$ directly follows from the defining constraint of the platform's optimal policy in (\ref{eq:e_hat-def}) and the continuity assumptions of the cost and harm functions. An interesting consequence of this result, however, is captured in the following proposition, where we effectively specify a strict lower bound on the public backlash as a necessary condition to induce the socially-optimal effort $e^*$.

\begin{restatable}{prop}{fourthprop}
\label{prop:ancillary-result}
The platform's optimal stable effort is guaranteed to be socially suboptimal unless the public standard becomes excessive by requiring effort $e_h > e^*$ if harm occurs.
\end{restatable}
Our ancillary result in Proposition \ref{prop:ancillary-result} captures the counterintuitive nature of the penalty scheme according to our model: it is not sufficient to set the ex post required effort to the optimal effort $e^*$, assuming $e^*$ were known; instead, to incentivize optimal effort, the public standard must overreact and mandate excessive, suboptimal effort $e_h > e^*$ as a response to any harmful event.

\subsection{Incentivizing Socially-Optimal Effort Under a Robust Public Standard}

\graf{Description of simpler setting}
Our descriptive model places an emphasis on overreacting to harmful events in order to incentivize socially-optimal effort. But since mandating suboptimal effort via such an overreaction is undesirable, we consider a simpler problem setting: suppose that the regulator has access to the platform's proprietary model and its underlying data, which can now be used as the public standard robust to performativity. The regulator thus has knowledge of the harm function $h$. Suppose further that the societal costs of harm $D$ are also given; the only missing information is the cost function $c$, or the platform's costs of effort to mitigate disinformation.

\begin{restatable}{prop}{fifthprop}
\label{prop:impossbility-result}
There is no way of adjusting the effort $e_c$ required by the public standard, purely as a function of the harm function $h$ and the cost of damages $D$, without regard to the cost function $c$, such that the platform is always incentivized to exert the socially-optimal level of effort.
\end{restatable}

\graf{Implications of impossibility result}
This result shows that even if a regulator has precise control over the public standard, without knowledge of the platform's costs, there is no way to set up the public standard's effort threshold such that the platform's individually-optimal effort level is always socially optimal. Thus, since social platforms' costs of precautionary effort underpin the incentive problem, it is crucial to model these costs in more detail to better understand their incentives relating to the control of disinformation. Determining how engagement translates to money, therefore, serves as an important avenue for future exploration, as platforms risk losing out on engagement revenue with content moderation.

\section{Conclusions}
\label{sec:conclusion}
Events like the Covid-19 ``infodemic'' or the Capitol Hill riots are recent examples of the harm associated with disinformation.  There is increasing evidence that the failure of social media platforms to control the spread of disinformation is due to incentive issues rather than a lack of technical ability \cite{nyt-fb, fb-whistleblower, frank_law}.
This work provides a formal analysis of these incentive issues that adapts the standard principal-agent framework to incorporate the unique features of the domain. Our formal model, although stylized, includes what we take to be key aspects of the setting, including the performativity of disinformation classification and public backlash as a response to harmful events.  Our formal results provide insights for the effective regulation of social media platforms.

\graf{Main results}
We argue that although a strict liability standard would theoretically align the platform's incentives with those of society, it is unlikely to be practical given the difficulty of assigning responsibility for harmful events to specific instances of disinformation after the fact.
Using our formal model, we derive a number of results relating to the use of a negligence standard.
Most importantly, we show that in the absence of a public backlash to harmful events, there is no monitoring scheme that can induce the platforms to perform a socially-optimal level of control of disinformation. However, a platform can be incentivized to exert more diligent effort than explicitly required by the regulator when the public standard of effort required to mitigate disinformation --- in terms of specifying what content ought to be limited from platforms --- includes the possibility of becoming excessive, and thus socially suboptimal, in the form of an overreaction in response to a harmful event.

\graf{Explicitly call out the undesirable results; Summarise impossibility result}
Clearly these results exhibit undesirable properties. Regulation of platforms via mandating excessive content moderation is not a practical recommendation. Furthermore, our impossibility result (Proposition \ref{prop:impossbility-result}) captures another undesirable property: even if a regulator possesses the same technical expertise and resources as a social platform, there is no way to induce the platform to control disinformation adequately via our mode of ex ante negligence regulation, without knowledge of the platform's costs of content moderation efforts.
\graf{Segue to prescriptive proposal}
Despite these perhaps unenviable conclusions, our modeling exercise offers valuable insights into the incentive issues relevant to platforms' control of online disinformation. Moreover, our results provide a lens through which further regulatory prescriptions for controlling disinformation might be derived.

\graf{Punchline}
Disinformation is one of the most urgent problems facing society. But it is a problem driven by incentives as much as by technology. This work takes a first step toward explicitly modeling the incentive issues that must be accounted for by any effective solution to the problem.


\graf{Limitations}
\subsection{Future Work}
Our model makes a number of simplifying assumptions. Treating public standards as explicit implies that a platform can guarantee a given probability of escaping punishment if it conforms to an explicit standard, which is an oversimplification of reality. The assumption that the platform can perfectly tune its proprietary model to flag toxic content is also unrealistic; technical challenges, although they may not pose the main obstacle to the practical control of disinformation, are nevertheless a real issue \cite{hao_2021}. Extending the model to more richly model these aspects are important directions for future work.

\graf{Binary Harm is wlog}
\paragraph{Homogeneous Harm} Another simplifying assumption of our setup the expression of the harm from disinformation as a binary event. This binary notion of harm might seem restrictive, especially because the harm from disinformation can manifest in many forms: rare events such as the Capitol Hill riots or the Pizzagate shooting are dramatic and immediately observable, in comparison to harm from the degradation of public discourse or from the spread of climate change denial or anti-vaccine propaganda, which are more subtle manifestations.

Regardless, our setup is without loss of generality: recall that our social welfare expression (Equation \ref{eq:social-welfare}) quantifies the expected societal costs of harm from disinformation, $h(e)D$ should a harmful event occur for a platform's given level of effort. This expression can be augmented to capture different types of harm: we will simply substitute our harm function with different probability distribution functions for the different kinds of harm and include the associated societal costs. This practice straightforwardly preserves our model's notion of quantifying expected harm.

\graf{Heterogenous Content}
\paragraph{Heterogeneous Content} While our notion of measuring content harmfulness via a binary harmful event is without loss of generality, it is meaningfully different to consider the heterogeneity of content in terms of how harmful a particular piece of content is and how much benefit it brings to a social platform. Our simple model of the platform's costs of effort $c(e)$ implies the homogeneity of all content with respect to the value it brings to the platform, since it suggests each content item attains the same amount of engagement from users.

Yet, in reality, just as content is not homogeneous in terms of the varying degrees of harmfulness of each item, content will also differ in the levels of user engagement attained. Therefore, modeling this heterogeneous relationship of the harm and benefit of content in future work will likely drive different conclusions. For example, with such explicit modeling, one question we might hope to answer is whether highly toxic content is more likely to produce high levels of user engagement (in the form of likes, shares, retweets, comments etc.) than less toxic content, thereby being more valuable to the platform. This will shed light on the degree to which the incentives of social platforms relating to the control of disinformation are misaligned with those of society, which, in turn, will inform the nature of any regulatory interventions required to realign these incentives.

\graf{Taxing toxicity}
\paragraph{Taxing toxicity} A \textit{Pigouvian tax} is a tax on a market transaction that generates a negative externality borne by individuals not directly involved in the transaction \cite{pig-tax}. Social platforms exhibit the precise criterion of generating negative externalities that calls for the levying of this tax: the more users a platform has, the more lucrative it is for advertisers to pay for the platform's services to target them with ads; and furthermore, the more time these users spend engaging with other users and content on the platform, the greater the opportunity for the platform to cater to the precise needs of advertisers. Thus, a platform benefits from more engagement than less, irrespective of whether such engagement is induced from harmful or benign content. But because it is only society that incurs the costs of harmful content, the idea behind taxation is to internalize the costs of toxicity to the original transaction between the platform and an advertiser.

Devising a taxation scheme requires a good harm model to measure content toxicity, using an access to data and expertise that is only available to platforms. Therefore, a mechanism designer (regulator) might instead impose taxation in a more crude manner. For example, the regulator can ask the platform report its cost function for moderating content and then tax the platform based on its report. Since the indirect costs of effort essentially capture the value of engagement, this mechanism levies a tax on user engagement on the platform more generally, rather than on the harmfulness of hosted content, which is the entity we wish to control on social platforms.

As shown by Proposition \ref{prop:impossbility-result}, the platform's costs of effort underpin the incentive problem for disinformation mitigation;
thus, any effective mechanism must in some way be responsive to these costs.
Naturally, such a mechanism must also factor in incentives that might prevent the platform from misreporting its true cost function for moderating content, in the hopes of attaining a lower tax rate, for instance.

\section{Ethical Considerations}
\label{sec:ethical-concerns}
Regulating social media is an especially sensitive issue.
Although allowing disinformation to spread unchecked is clearly unsustainable, disinformation control always runs the risk of becoming censorship. In this work, we take the existence of a ``public standard'' of acceptable postings for granted. However, the content of this public standard is a question of societal standards that can be settled only by public debate. Similarly, we analyze the use of ``monitoring'' without specifying its exact form. A naively implemented monitoring scheme would run the risk of serious privacy violations.

\newpage
\printbibliography

\newpage
\appendix

\section{Appendix}

Here we recall our main results and include proofs omitted from the main body of the paper.

\secondprop*
\begin{proof}
    By contradiction. Suppose the platform prefers to exert effort $e > e'$. Thus, the following must hold:
    \begin{align*}
        EU(e | e_c) &> EU(e' | e_c) \\
        \iff -c(e) - rP_f(e | e_c)F &> -c(e') - rP_f(e' | e_c)F \\
    \end{align*}
    By definition, $P_f(e' | e_c) = 0$ and therefore $P_f(e | e_c) = 0$. Thus,
    \begin{alignat*}{2}
    &&-c(e) - rP_f(e | e_c)F &> -c(e') - rP_f(e' | e_c)F \\
        &\iff& -c(e) &> -c(e') \\
        &\iff& c(e') &> c(e)
    \end{alignat*}
    which does not hold for $e > e'$ because by assumption $c'(e) > 0$ for all $e$ (contradiction).
\end{proof}

\firstlemma*
\begin{proof}
At $e_c$, the state-action value function for some effort $e$ is given by:
    \begin{align*}
        q_{\pi}(e_c, e) &= \E_{\pi}[R_{t+1} + \gamma v_{\pi}(S_{t+1}) \mid S_{t}=e_c, A_{t}=e] \\
        &= \sum_{e'_c}P(e'_c \mid s=e_c, a=e)[-c(e) + \gamma v_{\pi}(e'_c)] \\
        &= -c(e) + \gamma [h(e)v_{\pi}(e_h) + (1-h(e))(g(e_c)v_{\pi}(\chi(e_c))+(1-g(e_c))v_{\pi}(e_c))]
    \end{align*}
    By substituting in $d(\pi, e_c) = g(e_c)v_{\pi}(\chi(e_c))+(1-g(e_c))v_{\pi}(e_c)$ we have:
    \begin{equation}
    \label{eq:state-action-val}
       q_{\pi}(e_c, e) = -c(e) + \gamma [h(e)v_{\pi}(e_h) + (1-h(e))d(\pi, e_c)].
    \end{equation}
    Thus, for $q_{\pi}(e_c, e_2) > q_{\pi}(e_c, e_1)$, we have:
    \begin{alignat*}{2}
        && -c(e_2) + \gamma [h(e_2)v_{\pi}(e_h) + (1-h(e_2))d(\pi, e_c)] &>\\
        && -c(e_1) + \gamma [h(e_1)v_{\pi}(e_h) + (1-h(e_1))d(\pi, e_c)]\\
        &\iff&
           -c(e_2) + \gamma [h(e_2)v_{\pi}(e_h) + d(\pi, e_c) -h(e_2)d(\pi, e_c)] &>\\
        && -c(e_1) + \gamma [h(e_1)v_{\pi}(e_h) + d(\pi, e_c) -h(e_1)d(\pi, e_c)]\\
        &\iff&
        -c(e_2) + \gamma h(e_2)v_{\pi}(e_h) + \gamma d(\pi, e_c) -\gamma h(e_2)d(\pi, e_c) &>\\
        &&
        -c(e_1) + \gamma h(e_1)v_{\pi}(e_h) + \gamma d(\pi, e_c) -\gamma h(e_1)d(\pi, e_c)\\
        &\iff&
        -c(e_2) + \gamma h(e_2)v_{\pi}(e_h) -\gamma h(e_2)d(\pi, e_c) &>\\
        &&
        -c(e_1) + \gamma h(e_1)v_{\pi}(e_h) -\gamma h(e_1)d(\pi, e_c) \\
        &\iff&
        -c(e_2) - \gamma h(e_2)(d(\pi, e_c) -v_{\pi}(e_h)) &>\\
        &&
        -c(e_1) - \gamma h(e_1)(d(\pi, e_c) -v_{\pi}(e_h)) \\
        &\iff&
        \gamma h(e_1)(d(\pi, e_c) -v_{\pi}(e_h)) - \gamma h(e_2)(d(\pi, e_c) -v_{\pi}(e_h)) &> c(e_2) - c(e_1) \\
        &\iff&
        \gamma (h(e_1) - h(e_2))(d(\pi, e_c) -v_{\pi}(e_h)) &> c(e_2) - c(e_1) \\
        &\iff&
        d(\pi, e_c) - v_{\pi}(e_h) &> \frac{c(e_2) - c(e_1)}{\gamma (h(e_1) - h(e_2))}.
        \tag*{\qedhere}
    \end{alignat*}
\end{proof}

\secondlemma*
\begin{proof}
    The state value function for some arbitrary $e_c \le \tau$ is given by,
    \begin{equation}
        \label{eq:tau-state-val}
        v_{\pi^\tau}(e_c) = -c(\tau) + \gamma [h(\tau)v_{\pi^\tau}(e_h) + (1-h(\tau))(g(e_c)v_{\pi^\tau}(\chi(e_c))+(1-g(e_c))v_{\pi^\tau}(e_c))].
    \end{equation}
    
    Note that the platform's policy specifying effort for all $e_c \le \tau$ is fixed by definition; that is, $\pi^\tau(e_c) = \tau$ for all $e_c \le \tau$. Thus, the transition to state $e_h$ is also fixed because the transition probability $h(\tau)$ is fixed. And similarly, the probability that harm does not occur is also fixed at $(1-h(\tau))$.

Let $e_0 = \min \mathcal{S}$.    
We prove inductively that $v_{\pi^\tau}(e_k) = v_{\pi^\tau}(e_0)$ for all $e_0 \le e_k \le \tau$.  The base case ($v_{\pi^\tau}(e_0) = v_{\pi^\tau}(e_0)$) is immediate.  For the inductive step, assume that
$v_{\pi^\tau}(e_{k-1}) = v_{\pi^\tau}(e_0)$.  Then
\begin{align*}
  v_{\pi^\tau}(e_k)
  &= -c(\tau) + \gamma[h(\tau)v_{\pi^\tau}(e_h) + 
  (1-h(\tau))(1-g(e_k))v_{\pi^\tau}(e_k) + (1-h(\tau))g(e_k)v_{\pi^\tau}(e_{k-1})] \\
  &= -c(\tau) + \gamma[h(\tau)v_{\pi^\tau}(e_h) + 
  (1-h(\tau))(1-g(e_k))v_{\pi^\tau}(e_k) + (1-h(\tau))g(e_k)v_{\pi^\tau}(e_0)].
\end{align*}
Thus, $v_{\pi^\tau}(e_k) = g(e_k)V_1 + (1-g(e_k))V_0$, where
\begin{align*}
  V_1
  &= -c(\tau) + \gamma[h(\tau)v_{\pi^\tau}(e_h) + 
  (1-h(\tau))v_{\pi^\tau}(e_{k-1})]\\
  &= -c(\tau) + \gamma[h(\tau)v_{\pi^\tau}(e_h) + 
  (1-h(\tau))v_{\pi^\tau}(e_0)]\\
  &= v_{\pi^\tau}(e_0)
\end{align*}
and
\begin{align*}
V_0
&= -c(\tau) + \gamma h(\tau)v_{\pi^\tau}(e_h) + 
              \gamma(1-h(\tau))v_{\pi^\tau}(e_k) \\
&= -c(\tau) + \gamma h(\tau)v_{\pi^\tau}(e_h) + \gamma(1-h(\tau))[g(e_k)V_1 + (1-g(e_k))V_0].
\end{align*}
Note that the following is also true for $V_1$:
\begin{align*}
V_1 &= -c(\tau) + \gamma[h(\tau)v_{\pi^\tau}(e_h) + 
  (1-h(\tau))v_{\pi^\tau}(e_0)] \\
  &= -c(\tau) + \gamma[h(\tau)v_{\pi^\tau}(e_h) + 
  (1-h(\tau))g(e_k)v_{\pi^\tau}(e_0) + (1-h(\tau))(1-g(e_k))v_{\pi^\tau}(e_0)]\\
  &= -c(\tau) + \gamma h(\tau)v_{\pi^\tau}(e_h) + 
  \gamma(1-h(\tau))g(e_k)v_{\pi^\tau}(e_0) + \gamma(1-h(\tau))(1-g(e_k))v_{\pi^\tau}(e_0)\\
  &= \sum_{j=0}^\infty \gamma^j(1-h(\tau))^j(1-g(e_k))^j[-c(\tau) + \gamma h(\tau)v_{\pi^\tau}(e_h) + \gamma (1-h(\tau))g(e_k)v_{\pi^\tau}(e_0)] \\
  &= v_{\pi^\tau}(e_0)
\end{align*}
for all $g(e_k) \in [0,1]$.

Thus, for $V_0$:
\begin{align*}
V_0
&= -c(\tau) + \gamma h(\tau)v_{\pi^\tau}(e_h) + 
              \gamma(1-h(\tau))[g(e_k)V_1 + (1-g(e_k))V_0]\\
&= -c(\tau) + \gamma h(\tau)v_{\pi^\tau}(e_h) + \gamma(1-h(\tau))[g(e_k)v_{\pi^\tau}(e_0) + (1-g(e_k))V_0]\\
&= -c(\tau) + \gamma h(\tau)v_{\pi^\tau}(e_h) + \gamma(1-h(\tau))g(e_k)v_{\pi^\tau}(e_0) + \gamma(1-h(\tau))(1-g(e_k))V_0\\
&= \sum_{j=0}^\infty \gamma^j(1-h(\tau))^j(1-g(e_k))^j[-c(\tau) + \gamma h(\tau)v_{\pi^\tau}(e_h) + \gamma (1-h(\tau))g(e_k)v_{\pi^\tau}(e_0)] \\
&= V_1 \\
&= v_{\pi^\tau}(e_0).
\end{align*}
But then
\begin{align*}
  v_{\pi^\tau}(e_k)
  &= g(e_k)V_1 + (1-g(e_k))V_0 \\
  &= g(e_k)v_{\pi^\tau}(e_0) + (1-g(e_k))v_{\pi^\tau}(e_0) \\
  &= v_{\pi^\tau}(e_0)
\end{align*}
for all $g(e_k) \in [0,1]$, and we are done.
\end{proof}

\thirdprop*
\begin{proof}
    For ease of notation, let $S=\{e^0, e^1, \ldots, e^h\}$ denote the set of all states with $e^0 < \cdots < e^h$, and let $\pi = \pi^\tau$ with $\tau = 0$. Note that this specification is w.l.o.g.; for $\tau > 0$, we will consider a subset of $S$ such that the first state of this subset $e^0 = \sup\{ e \in S \mid e \le \tau \}$, since from Lemma \ref{lem:tau-states} we know that the state value function for all $e \le \tau$ is fixed.
    
    Now we move on to the proof. Suppose the claim is false. Then $\{e \mid v_{\pi}(e) < v_{\pi}(e^h)\} \ne \varnothing$. Let $e^z = \min\{e \mid v_{\pi}(e) < v_{\pi}(e^h)\}$ and $d(e^j) = (1-g(e^j))v_\pi(e^j) + g(e^j)v_\pi(e^{j-1})$ for all $0 \le j \le h$.

    First, observe that
    \begin{align*}
      d(e^z) &= (1-g(e^z))v_\pi(e^z) + g(e^z)v_\pi(e^{z-1}) \\
             &\ge (1-g(e^z))v_\pi(e^z) + g(e^z)v_\pi(e^z) \\
             &= v_\pi(e^z),
    \end{align*}
    where the inequality follows from combining the assumptions $v_\pi(e^h) > v_\pi(e^z)$ with
    $v_\pi(e^{z-1}) \ge v_\pi(e^h)$, both from the definition of $e^z$.  Note that if $e^z = e^0$, then the same result holds, since $g(e^0) = 0$.
    
    It then follows that
    \begin{align*}
      v_\pi(e^z) &=   -c(e^z) + \gamma [ h(e^z)v_\pi(e^h) + (1-h(e^z))d(e^z) ] \\
             &\ge -c(e^z) + \gamma [ h(e^z)v_\pi(e^h) + (1-h(e^z))v_\pi(e^z) ] \\
             &>   -c(e^z) + \gamma [ h(e^z)v_\pi(e^z) + (1-h(e^z))v_\pi(e^z) ] \\
             &=   -c(e^z) + \gamma v_\pi(e^z) \\
             &\ge \sum_{j=0}^\infty \gamma^j(-c(e^z)).
    \end{align*}
    
    We now show inductively that $v_\pi(e^k) \ge v_\pi(e^h)$ for all $z \le k < h$.
    
    The base case is $e^{h-1}$.  Suppose the contrary that $v_\pi(e^{h-1}) < v_\pi(e^h)$.  Then we have
    \begin{align*}
      d(e^h) &= (1-g(e^h))v_\pi(e^h) + g(e^h)v_\pi(e^{h-1})\\
             &\le v_\pi(e^h), 
    \end{align*}
    because $0 \le g(e^h) \le 1$, which gives
    \begin{align*}
      v_\pi(e^h) &=   -c(e^h) + \gamma [ h(e^h)v_\pi(e^h) + (1-h(e^h))d(e^h) ] \\
             &\le -c(e^h) + \gamma [ h(e^h)v_\pi(e^h) + (1-h(e^h))v_\pi(e^h) ] \\
             &=   -c(e^h) + \gamma v_\pi(e^h) \\
             &\le \sum_{j=1}^\infty \gamma^j(-c(e^h)) \\
             &<   \sum_{j=1}^\infty \gamma^j(-c(e^z)) \\
             &<   v_\pi(e^z),
    \end{align*}
    contradicting the definition of $e^z$.

    For the inductive step, assume that $v_\pi(e^k) \ge v_\pi(e^h)$, for some $z < k < h$.  Then we show that $v_\pi(e^{k-1}) \ge v_\pi(e^h)$.  Assume not; then similarly we have
    \begin{align*}
      d(e^k) &=   (1-g(e^k))v_\pi(e^k) + g(e^k)v_\pi(e^{k-1}) \\
             &\le (1-g(e^k))v_\pi(e^k) + g(e^k)v_\pi(e^h) \\
             &\le (1-g(e^k))v_\pi(e^k) + g(e^k)v_\pi(e^k) \\
             &=   v_\pi(e^k)
    \end{align*}
    and thus
    \begin{align*}
      v_\pi(e^k) &=   -c(e^k) + \gamma [ h(e^k)v_\pi(e^h) + (1-h(e^k))d(e^k) ] \\
             &\le -c(e^k) + \gamma [ h(e^k)v_\pi(e^h) + (1-h(e^k))v_\pi(e^k) ] \\
             &\le -c(e^k) + \gamma [ h(e^k)v_\pi(e^k) + (1-h(e^k))v_\pi(e^k) ] \\
             &=   -c(e^k) + \gamma v_\pi(e^k) \\
             &\le \sum_{j=1}^\infty \gamma^j(-c(e^k)) \\
             &<   \sum_{j=1}^\infty \gamma^j(-c(e^z)) \\
             &<   v_\pi(e^z) \\
             &<   v_\pi(e^h) \\
             &\le v_\pi(e^k),
    \end{align*}
    again yielding a contradiction.
    
    Therefore, $v_\pi(e^k) \ge v_\pi(e^h)$ is true for all $z \le k < h$, which in particular implies that the initial claim $\{e \mid v_{\pi}(e) < v_{\pi}(e^h)\} \ne \varnothing$ must be false, thus completing the proof.
\end{proof}

\mainthm*
\begin{proof}
    By contradiction. Suppose $\pi^{\hat{e}}$ is suboptimal. Then by the process of policy improvement, there must exist a state $e_c$ where some effort $e \ne \pi^{\hat{e}}(e_c)$ guarantees a higher expected reward than $\pi^{\hat{e}}(e_c)$. Thus, we apply the policy improvement theorem (Lemma \ref{lem:policy-improvement}) to find any such $e_c$ where $q_{\pi^{\hat{e}}}(e_c, e) > v_{\pi^{\hat{e}}}(e_c)$ holds, which would imply that a greedy deviation from $\pi^{\hat{e}}$ exists as the better policy.
    
    \paragraph{Case 1 ($\boldsymbol{\forall e_c}$): Less aggressive effort than $\boldsymbol{e_c}$}
    The first deviation from $\pi^{\hat{e}}$ at any $e_c$ might be to exert less aggressive effort $e < e_c$. Suppose that less aggressive effort $e$ guarantees a higher expected reward than the required effort $e_c$. However, we know that lower effort than $e_c$ does not guarantee a higher expected reward for all $e_c$ because $e_c$ by definition is the platform's individually-optimal level of effort.
    Thus, we have a contradiction and this deviation does not work.
    
    \paragraph{Case 2 ($\boldsymbol{e_c \le \hat{e}}$): Less aggressive effort than $\boldsymbol{\hat{e}}$} Suppose that the platform prefers to exert less aggressive effort $e_1$ such that $e_c \le e_1 < \hat{e}$. Then $q_{\pi^{\hat{e}}}(e_c, e_1) > q_{\pi^{\hat{e}}}(e_c, \hat{e})$ must be true.
    
    Thus, $q_{\pi^{\hat{e}}}(e_c, \hat{e}) > q_{\pi^{\hat{e}}}(e_c, e_1)$ \emph{must not} be true (contrapositive); or, by substituting in equation (\ref{eq:thresh-preference}) from Lemma \ref{lem:aggressive-filtering}, the following must not hold:
    \begin{equation}
        \label{eq:e_hat-suff-condition}
        d(\pi^{\hat{e}}, e_c) - v_{\pi^{\hat{e}}}(e_h) > \frac{c(\hat{e}) - c(e_1)}{\gamma (h(e_1) - h(\hat{e}))}.
    \end{equation}
    
    From the definition in (\ref{eq:e_hat-def}), note that because $\hat{e}$ is the supremum taken over a closed interval, it satisfies the following equation (\emph{intermediate value theorem}):
    \begin{equation}
        \label{eq:e_hat}
        q_{\pi^{\hat{e}}}(s^{-1}(\hat{e}), \hat{e}) - q_{\pi^{\hat{e}}}(e_h, e_h) = -\frac{c'(\hat{e})}{\gamma h'(\hat{e})}.
    \end{equation}
    Now consider the L.H.S of (\ref{eq:e_hat-suff-condition}) and of (\ref{eq:e_hat}). Recall that $d(\pi^{\hat{e}}, e_c) = g(e_c)v_{\pi^{\hat{e}}}(\chi(e_c)) + (1-g(e_c))
    v_{\pi^{\hat{e}}}(e_c)$. Since $\pi^{\hat{e}}(e_c) = \hat{e}$ is fixed for all $e_c < \hat{e}$, the value functions $v_{\pi^{\hat{e}}}(\chi(e_c))$ and $v_{\pi^{\hat{e}}}(e_c)$ must be equal (Lemma \ref{lem:tau-states}). Thus, $d(\pi^{\hat{e}}, e_c) = v_{\pi^{\hat{e}}}(e_c)$ as $0 \le g(e_c) \le 1$. Furthermore, because $s^{-1}(\hat{e}) < \hat{e}$ by definition, $q_{\pi^{\hat{e}}}(s^{-1}(\hat{e}),\hat{e})=v_{\pi^{\hat{e}}}(e_c)$ must be true. Moreover, $v_{\pi^{\hat{e}}}(e_h) = q_{\pi^{\hat{e}}}(e_h, e_h)$ as $e_h \ge \hat{e}$. Thus, the L.H.S of (\ref{eq:e_hat-suff-condition}) and of (\ref{eq:e_hat}) are equal, or
    \begin{equation}
        \label{eq:e_hat-convenient-equivalence}
        d(\pi^{\hat{e}}, e_c) - v_{\pi^{\hat{e}}}(e_h) = q_{\pi^{\hat{e}}}(s^{-1}(\hat{e}), \hat{e}) - q_{\pi^{\hat{e}}}(e_h, e_h).
    \end{equation}
    
    Suppose that the following is true of the R.H.S of (\ref{eq:e_hat-suff-condition}) and (\ref{eq:e_hat}):
    \begin{equation}
        \label{eq:ratio-marginals}
        -\frac{c'(\hat{e})}{\gamma h'(\hat{e})} > \frac{c(\hat{e}) - c(e_1)}{\gamma (h(e_1) - h(\hat{e}))}
    \end{equation}
    Thus,
    \begin{alignat*}{2}
        &&-\frac{c'(\hat{e})}{\gamma h'(\hat{e})} &> \frac{c(\hat{e}) - c(e_1)}{\gamma (h(e_1) - h(\hat{e}))}  \\
        &\iff& -\frac{c'(\hat{e})}{h'(\hat{e})} &> \frac{c(\hat{e}) - c(e_1)}{h(e_1) - h(\hat{e})} \\
        &\iff& -\frac{c'(\hat{e})(e_1 - \hat{e})}{h'(\hat{e})(e_1 - \hat{e})} &> -\frac{c(e_1) - c(\hat{e})}{h(e_1) - h(\hat{e})} \\
        &\iff& \frac{c(e_1) - c(\hat{e})}{h(e_1) - h(\hat{e})} &> \frac{c'(\hat{e})(e_1 - \hat{e})}{h'(\hat{e})(e_1 - \hat{e})}
    \end{alignat*}
    Notice that the final inequality is always true: we know by assumption that $c$ is strictly convex ($c''(e) > 0$) and so from equation (\ref{eq:strict-convexity-condition}) in Lemma \ref{lem:fo-condition-convexity} it follows that the numerator of the L.H.S must be strictly greater than the numerator of the R.H.S, i.e., $c(e_1) - c(\hat{e}) > c'(\hat{e})(e_1 - \hat{e})$; similarly, because $h$ is convex ($h''(e) \ge 0$), the denominator of the L.H.S must be be weakly greater than the denominator of the R.H.S, i.e., $h(e_1) - h(\hat{e}) \ge h'(\hat{e})(e_1 - \hat{e})$. Since $h'(e) < 0$ and $e_1 < \hat{e}$, it follows that the L.H.S fraction overall is strictly greater (\emph{less negative}) than the R.H.S fraction (\emph{more negative}).
    
    Therefore, if (\ref{eq:ratio-marginals}) holds, then condition (\ref{eq:e_hat-suff-condition}) must also hold because:
    \begin{align*}
        q_{\pi^{\hat{e}}}(s^{-1}(\hat{e}), \hat{e}) - q_{\pi^{\hat{e}}}(e_h, e_h) &= -\frac{c'(\hat{e})}{\gamma h'(\hat{e})} \\
        \iff d(\pi^{\hat{e}}, e_c) - v_{\pi^{\hat{e}}}(e_h) &= -\frac{c'(\hat{e})}{\gamma h'(\hat{e})} \\
        &> \frac{c(\hat{e}) - c(e_1)}{\gamma (h(e_1) - h(\hat{e}))}.
    \end{align*}
    If (\ref{eq:e_hat-suff-condition}) holds, the contrapositive statement is false, and so the original statement must also be false; thus, the platform instead prefers to exactly exert effort $\hat{e}$, and no less, for all $e_c < \hat{e}$, a contradiction.
    
    \paragraph{Case 3 ($\boldsymbol{e_c \le \hat{e}}$): More aggressive effort than $\boldsymbol{\hat{e}}$} Suppose the platform prefers to exert excessive effort at some $e_2 > \hat{e}$. It follows that $q_{\pi^{\hat{e}}}(e_c, e_2) > q_{\pi^{\hat{e}}}(e_c, \hat{e})$ must hold, and so we have (Lemma \ref{lem:aggressive-filtering}):
    \begin{equation}
        \label{eq:e_hat-necc-condition}
        d(\pi^{\hat{e}}, e_c) - v_{\pi^{\hat{e}}}(e_h) > \frac{c(e_2) - c(\hat{e})}{\gamma (h(\hat{e}) - h(e_2))},
    \end{equation}
    must also hold.
    Recall from (\ref{eq:e_hat-convenient-equivalence}) that,
    \[ d(\pi^{\hat{e}}, e_c) - v_{\pi^{\hat{e}}}(e_h) = q_{\pi^{\hat{e}}}(s^{-1}(\hat{e}), \hat{e}) - q_{\pi^{\hat{e}}}(e_h, e_h). \]
    Thus,
    \begin{align*}
        d(\pi^{\hat{e}}, e_c) - v_{\pi^{\hat{e}}}(e_h) > \frac{c(e_2) - c(\hat{e})}{\gamma (h(\hat{e}) - h(e_2))} \\
        \iff q_{\pi^{\hat{e}}}(s^{-1}(\hat{e}), \hat{e}) - q_{\pi^{\hat{e}}}(e_h, e_h) > \frac{c(e_2) - c(\hat{e})}{\gamma (h(\hat{e}) - h(e_2))}
    \end{align*}
    We know from (\ref{eq:e_hat}) that,
    \[ q_{\pi^{\hat{e}}}(s^{-1}(\hat{e}), \hat{e}) - q_{\pi^{\hat{e}}}(e_h, e_h) = -\frac{c'(\hat{e})}{\gamma h'(\hat{e})}. \]
    Thus, in order to guarantee that (\ref{eq:e_hat-necc-condition}) holds, the following must be true:
    \begin{alignat*}{2}
        &&-\frac{c'(\hat{e})}{\gamma h'(\hat{e})} &> \frac{c(e_2) - c(\hat{e})}{\gamma (h(\hat{e}) - h(e_2))} \\
        &\iff& -\frac{c'(\hat{e})}{h'(\hat{e})} &> \frac{c(e_2) - c(\hat{e})}{h(\hat{e}) - h(e_2)} \\
        &\iff& -\frac{c'(\hat{e})(e_2 - \hat{e})}{h'(\hat{e})(e_2 - \hat{e})} &> -\frac{c(e_2) - c(\hat{e})}{h(e_2) - h(\hat{e})} \\
        &\iff& \frac{c(e_2) - c(\hat{e})}{h(e_2) - h(\hat{e})} &> \frac{c'(\hat{e})(e_2 - \hat{e})}{h'(\hat{e})(e_2 - \hat{e})}
    \end{alignat*}
    However, notice that this inequality \emph{does not} hold for $e_2 > \hat{e}$: since $c$ is strictly convex, we know from Lemma \ref{lem:fo-condition-convexity} that the numerator of the L.H.S is strictly greater than that of the R.H.S, i.e., $c(e_2) - c(\hat{e}) > c'(\hat{e})(e_2 - \hat{e})$; and similarly, because $h$ is convex, the denominator of the L.H.S is weakly greater than that of the R.H.S, i.e., $h(e_2) - h(\hat{e}) \ge h'(\hat{e})(e_2 - \hat{e})$. Since $h'(e) < 0$ and $\hat{e} < e_2$, it follows that the L.H.S fraction overall must be strictly smaller (\emph{more negative}) than the R.H.S fraction (\emph{less negative}), that is,
    \[ \frac{c(e_2) - c(\hat{e})}{h(e_2) - h(\hat{e})} < \frac{c'(\hat{e})(e_2 - \hat{e})}{h'(\hat{e})(e_2 - \hat{e})} \]
    must be true, a contradiction.
    
    \paragraph{Case 4 ($\boldsymbol{e_c > \hat{e}}$):
    More aggressive effort than $\boldsymbol{e_c}$} We prove an intermediate result to arrive at our contradiction for this case. We first show that $\hat{e}$ is the optimal effort threshold for all threshold strategies.
    
    Let $\tau$ be the smallest $\tau > \hat{e}$ satisfying $q_{\pi^\tau}(e,\pi^\tau(e)) \ge q_{\pi^{\hat{e}}}(e,\pi^{\hat{e}}(e))$ for all $e\in S$. Let $s^{-1}(\tau) = e_1 < \tau$. Then following the definition of $\hat{e}$ in (\ref{eq:e_hat-def}), we have
    \begin{alignat*}{2}
      &&       q_{\pi^\tau}(e_1,\tau) - v_{\pi^\tau}(e^h) &< \frac{-c'(\tau)}{\gamma h'(\tau)} \\
      &\iff&       q_{\pi^\tau}(e_1,\tau) - v_{\pi^\tau}(e^h) &< \frac{-c'(\sigma)}{\gamma h'(\sigma)}
    \end{alignat*}
    
    for $e_1 < \sigma < \tau$ and $\tau - \sigma$ sufficiently small. But since $d(\pi^{\tau},e_1) = q_{\pi^\tau}(e_1,\tau)$ (Lemma \ref{lem:tau-states}), we have
    \begin{align*}
      d(\pi^{\tau},e_1) - v_{\pi^\tau}(e^h) &< \frac{-c'(\sigma)}{\gamma h'(\sigma)} \\
                                       &< \frac{c(\tau) - c(\sigma)}{\gamma(h(\sigma)-h(\tau))},
    \end{align*}
    which implies by Lemma \ref{lem:aggressive-filtering} that $q_{\pi^\tau}(e_1,\sigma) \ge q_{\pi^\tau}(e_1,\tau)$, and hence by the policy improvement theorem, $v_{\pi^{\sigma}}(e) \ge v_{\pi^{\tau}}(e)$ for all $e \in S$, contradicting the definition of $\tau$. Hence there is no such threshold $\tau > \hat{e}$, and so $\hat{e}$ is the optimal threshold among all threshold strategies.
    
    Now suppose the platform prefers to exert more aggressive effort at $e_2 > e_c$ for some $e_c > \hat{e}$. Thus, by Lemma \ref{lem:aggressive-filtering}, the following must hold:
    
    \begin{alignat*}{2}
        &&d(\pi^{\hat{e}}, e_c) - v_{\pi^{\hat{e}}}(e_h) &> \frac{c(e_2) - c(e_c)}{\gamma (h(e_c) - h(e_2))} \\
        &\iff& g(e_c)v_{\pi^{\hat{e}}}(\chi(e_c)) + (1-g(e_c))v_{\pi^{\hat{e}}}(e_c) - v_{\pi^{\hat{e}}}(e_h) &> \frac{c(e_2) - c(e_c)}{\gamma (h(e_c) - h(e_2))}.
    \end{alignat*}
    Thus, we have:
    \begin{align*}
         g(e_c)v_{\pi^{\hat{e}}}(\chi(e_c)) + (1-g(e_c))v_{\pi^{\hat{e}}}(e_c) - v_{\pi^{\hat{e}}}(e_h) &> \frac{c(e_2) - c(e_c)}{\gamma (h(e_c) - h(e_2))} \\
        &> -\frac{c'(e_c)}{\gamma h'(e_c)} \\
        &> -\frac{c'(\hat{e})}{\gamma h'(\hat{e})} \\
        &= q_{\pi^{\hat{e}}}(s^{-1}(\hat{e}), \hat{e}) - q_{\pi^{\hat{e}}}(e_h, e_h) \\
        &= v_{\pi^{\hat{e}}}(s^{-1}(\hat{e})) - v_{\pi^{\hat{e}}}(e_h).
    \end{align*}
    Thus,
    \begin{alignat*}{2}
        &&      g(e_c)v_{\pi^{\hat{e}}}(\chi(e_c)) + (1-g(e_c))v_{\pi^{\hat{e}}}(e_c) - v_{\pi^{\hat{e}}}(e_h) &>   v_{\pi^{\hat{e}}}(s^{-1}(\hat{e})) - v_{\pi^{\hat{e}}}(e_h) \nonumber \\
        &\iff& g(e_c)v_{\pi^{\hat{e}}}(\chi(e_c)) + (1-g(e_c))v_{\pi^{\hat{e}}}(e_c) &> v_{\pi^{\hat{e}}}(s^{-1}(\hat{e})),
    \end{alignat*}
    which implies that $v_{\pi^{\hat{e}}}(e_c) > v_{\pi^{\hat{e}}}(s^{-1}(\hat{e}))$ and/or $v_{\pi^{\hat{e}}}(\chi(e_c)) > v_{\pi^{\hat{e}}}(s^{-1}(\hat{e}))$. Thus, it follows that a new threshold strategy with threshold strictly greater than $\hat{e}$ will be preferable to $\hat{e}$, since exerting more aggressive effort $e_2$ in state $e_c$ such that $e_2 > e_c > \hat{e}$ yields a better value. However, this implication contradicts our intermediate result because no threshold greater than $\hat{e}$ is optimal and we are done.
    
    The process of policy improvement must give us a strictly better policy except when the original policy is already optimal \cite{Sutton1998}.
    Since there exists no greedy deviation $e \ne \pi^{\hat{e}}(e_c)$ such that $q(e_c, e) > q(e_c, \pi^{\hat{e}}(e_c))$ is true for any $e_c$, the proposed policy $\pi^{\hat{e}}$ must be optimal, thus completing the proof.
\end{proof}

\existenceprop*
\begin{proof}
    We know from Theorem \ref{thm:form-optimal-policy} that under the specified conditions, the platform's optimal effort at any state $e_c$ is a threshold strategy with threshold $\tau = \hat{e}$. In order to induce $e^*$ as the optimal threshold, $\hat{e}$ must equal $e^*$; thus, from the defining constraint in (\ref{eq:e_hat-def}), there must exist some $e_h$ such that the following holds:
    \begin{align}
        q_{\pi^{e^*}}(s^{-1}(e^*), e^*) - q_{\pi^{e^*}}(e_h, \pi^{e^*}(e_h)) &= -\frac{c'(e^*)}{\gamma h'(e^*)} \nonumber \\
        \label{eq:e*_condition}
        \iff v_{\pi^{e^*}}(e_0) - v_{\pi^{e^*}}(e_h) &= -\frac{c'(e^*)}{\gamma h'(e^*)}.
    \end{align}
    where $e_0 = s^{-1}(e^*) \le e^*$ (by definition).
    
    Thus, we have
    \begin{align*}
        v_{\pi^{e^*}}(e_0) &= -c(e^*) + \gamma [h(e^*)v_{\pi^{e^*}}(e_h) + (1-h(e^*))v_{\pi^{e^*}}(e_0)] \\
        v_{\pi^{e^*}}(e_0) &= -c(e^*) + \gamma h(e^*)v_{\pi^{e^*}}(e_h) + \gamma (1-h(e^*))v_{\pi^{e^*}}(e_0) \\ 
        v_{\pi^{e^*}}(e_0) - \gamma (1-h(e^*))v_{\pi^{e^*}}(e_0) &= -c(e^*) + \gamma h(e^*)v_{\pi^{e^*}}(e_h) \\
        v_{\pi^{e^*}}(e_0)[1 - \gamma (1-h(e^*))] &= -c(e^*) + \gamma h(e^*)v_{\pi^{e^*}}(e_h),
    \end{align*}
    and finally
    \begin{equation}
        \label{eq:e0_equation}
        v_{\pi^{e^*}}(e_0) = \frac{-c(e^*) + \gamma h(e^*)v_{\pi^{e^*}}(e_h)}{1 - \gamma (1-h(e^*))}.
    \end{equation}
    
    By substituting (\ref{eq:e0_equation}) in (\ref{eq:e*_condition}), we have
    \begin{align*}
        v_{\pi^{e^*}}(e_0) - v_{\pi^{e^*}}(e_h) &= -\frac{c'(e^*)}{\gamma h'(e^*)} \\
        \frac{-c(e^*) + \gamma h(e^*)v_{\pi^{e^*}}(e_h)}{1 - \gamma (1-h(e^*))} - v_{\pi^{e^*}}(e_h) &= -\frac{c'(e^*)}{\gamma h'(e^*)} \\
        -c(e^*) + \gamma h(e^*)v_{\pi^{e^*}}(e_h) - (1 - \gamma (1-h(e^*)))v_{\pi^{e^*}}(e_h) &= -\frac{c'(e^*)(1 - \gamma (1-h(e^*)))}{\gamma h'(e^*)} \\
        -v_{\pi^{e^*}}(e_h)[-\gamma h(e^*) + 1 - \gamma (1-h(e^*))] &= c(e^*) -\frac{c'(e^*)(1 - \gamma (1-h(e^*)))}{\gamma h'(e^*)} \\
        -v_{\pi^{e^*}}(e_h)[-\gamma h(e^*) + 1 - \gamma + \gamma h(e^*))] &= c(e^*) -\frac{c'(e^*)(1 - \gamma (1-h(e^*)))}{\gamma h'(e^*)} \\
        -v_{\pi^{e^*}}(e_h)(1 - \gamma) &= c(e^*) -\frac{c'(e^*)(1 - \gamma (1-h(e^*)))}{\gamma h'(e^*)},
    \end{align*}
    and finally
    \begin{equation}
        \label{eq:e*_condition_final}
        -v_{\pi^{e^*}}(e_h) = (\frac{1}{1 - \gamma})(c(e^*) -\frac{c'(e^*)(1 - \gamma (1-h(e^*)))}{\gamma h'(e^*)}).
    \end{equation}
    
    We show by the \emph{intermediate value theorem} (IVT) that (\ref{eq:e*_condition_final}) holds for some $e_h \in (e_{min}, e_{max})$ where $e_{min}$ and $e_{max}$ correspond to the lowest and highest possible levels of effort, respectively.
    
    Let $G(e_h) = -v_{\pi^{e^*}}(e_h) - K$ where $K = (\frac{1}{1 - \gamma})(c(e^*) -\frac{c'(e^*)(1 - \gamma (1-h(e^*)))}{\gamma h'(e^*)})$. First, observe that $-v_{\pi^{e^*}}(e_h) \in (c(e_h),\frac{c(e_h)}{1-\gamma})$ by construction. Moreover, note that $K > 0$.
    Thus, we have the following at the lower bound of $e_h$:
    \begin{align*}
        G(e_{min}) &= -v_{\pi^{e^*}}(e_{min}) - K \\
        &< \frac{c(e_{min})}{1-\gamma} - (\frac{1}{1 - \gamma})(c(e^*) -\frac{c'(e^*)(1 - \gamma (1-h(e^*)))}{\gamma h'(e^*)}) \\
        &= \frac{c(e_{min})}{1-\gamma} - \frac{c(e^*)}{1-\gamma} + (\frac{1}{1 - \gamma})(\frac{c'(e^*)(1 - \gamma (1-h(e^*)))}{\gamma h'(e^*)}) \\
        &< 0,
    \end{align*}
    for all $e^* \ge e_{min}$.
    
    And we have the following at the upper bound of $e_h$:
    \begin{align*}
        G(e_{max}) &= -v_{\pi^{e^*}}(e_{max}) - K \\
        &> c(e_{max}) - K \\
        &> 0,
    \end{align*}
    which holds because the cost of exerting maximum possible effort, or $c(e_{max})$, is sufficiently large (by assumption).
    
    Hence, because $G(e_{min}) < 0 < G(e_{max})$, it follows by the IVT that there exists some $e_h \in (e_{min},e_{max})$ such that
    \begin{alignat*}{2}
        &&G(e_h) &= 0 \\
        &\iff& -v_{\pi^{e^*}}(e_h) - K &= 0 \\
        &\iff& -v_{\pi^{e^*}}(e_h) &= K \\
        &\iff& -v_{\pi^{e^*}}(e_h) &= (\frac{1}{1 - \gamma})(c(e^*) -\frac{c'(e^*)(1 - \gamma (1-h(e^*)))}{\gamma h'(e^*)}),
    \end{alignat*}
    and we are done.
\end{proof}

\fourthprop*
\begin{proof}
    This result directly follows from the defining condition of the platform's stable effort $\hat{e}$ in (\ref{eq:e_hat-def}). For $\hat{e}$ to equal $e^*$, the required effort $e_h$ must be strictly greater than $e^*$. By definition, $\hat{e}$ is the supremum over the closed interval $[0,e_h]$ and so if $e_h < e^*$, then $\hat{e} < e^*$ is also true.
    
    If $e_h = e^*$, then $\hat{e} < e^*$ is also true; the L.H.S of (\ref{eq:e_hat-def}) equals zero for $e = e_h$, or
    \[ q_{\pi^{e_h}}(s^{-1}(e_h), e_h) - q_{\pi^{e_h}}(e_h, \pi^{e_h}(e_h)) = 0, \]
    while the R.H.S is always positive, or \[ -\frac{c'(e_h)}{\gamma h'(e_h)} > 0, \]
    since $c'(e) > 0$ and $h'(e) < 0$ for all $e \in [0,1]$, and therefore the inequality is not satisfied.
    Thus, $e_h > e^*$ must be true in order for the platform's stable effort $\hat{e}$ to equal $e^*$.
\end{proof}

\fifthprop*
\begin{proof}
Consider the simplest possible case where we assume there exist only two possible cost functions, $c_1$ and $c_2$. Let $e^*_1$ be the socially-optimal effort induced by cost function $c_1$ and $e^*_2$ be that induced by cost function $c_2$, and let $e^*_2 > e^*_1$. Note that the $e^*_1$ and $e^*_2$ can be trivially computed from equation (\ref{eq:social-welfare}) by equating the marginal social welfare to zero.

Suppose that the actual socially-optimal effort is $e^*_1$. Note first that if the regulator sets the public standard to require effort $e_c = e^*_1$, there should not be any increase in this level of effort because a transition to some new $e_h > e_c$ will mandate excessive and therefore suboptimal effort as the platform at least follows the public standard's specified effort (by assumption). Thus, if the public standard is set to require effort $e^*_1$, we are done.

However, now suppose that the actual socially-optimal effort is $e^*_2$ and the public standard currently specifies $e^*_1$ as the required effort. Then, an increase in effort to $e^*_2$ is necessary to incentivize the platform to exert the socially-optimal effort (contradiction).

Similarly, a decrease in the effort required by the public standard does not guarantee that platform is always induced to exert the socially-optimal effort level: If $e^*_2$ is socially-optimal, then a decrease in the required effort to some $e_c < e^*_2$ does not guarantee that the platform will continue to exert effort at $e^*_2$; the platform may exert less and therefore socially suboptimal effort at $e_c$, since without increasing the required effort, there is no way for a static public standard to induce the platform to exert more than the prescribed effort as shown in Proposition \ref{prop:bare-minimum-effort}. Therefore, if the public standard is set to require $e^*_1$, the platform is no longer guaranteed to exert the socially-optimal $e^*_2$.

Thus, because no adjustment to the public standard's required effort works, the regulator needs to know whether the platform's true cost function is $c_1$ or $c_2$ to incentivize the socially-optimal effort at all times. And since there exist more than just two possible choices for the platform's actual cost function, there is no way for a regulator to guarantee that the platform exerts socially optimal effort for mitigating disinformation with any increasing or decreasing adjustments to the public standard's required effort.
\end{proof}

\end{document}